\documentclass[twocolumn]{IEEEtran}
\usepackage{url,amsmath,graphicx}
 \usepackage{stfloats}
 \usepackage{eqparbox}
  \usepackage{array}
  \usepackage{fixltx2e}
  \usepackage{stfloats}
   \usepackage{cite}
   \usepackage{url}
\usepackage{fancyhdr}
\usepackage{cite}
\usepackage{graphicx}
\usepackage{psfrag}
\usepackage{url}
\usepackage{stfloats}
\usepackage{amsmath}
\usepackage{array}
\usepackage{fancyhdr}
\usepackage{hyperref}
\usepackage{psfrag}
\usepackage{subcaption}
\usepackage{graphicx}
\usepackage{amssymb}
\usepackage{color}
\usepackage{url,amsmath,graphicx}
\usepackage{amsmath,amssymb,amscd,latexsym,dsfont,amsthm}
 \usepackage{stfloats}
 \usepackage{eqparbox}
  \usepackage{array}
  \usepackage{fixltx2e}
  \usepackage{stfloats}
   \usepackage{cite}
   \usepackage{url}
\usepackage{cite}
\usepackage{graphicx}
\usepackage{psfrag}
\usepackage{url}
\usepackage{stfloats}
\usepackage{array}
\usepackage{fancyhdr}
\usepackage{psfrag}
\usepackage{subcaption}
\usepackage{graphicx}
\usepackage{lipsum}
\usepackage{amssymb}
\usepackage{color}
\usepackage{mathtools, cuted}
\usepackage{lipsum, color}
\usepackage{amssymb}
\usepackage{color}

\newtheorem{theorem}{Theorem}
\newtheorem{lemma}{Lemma}
\newtheorem{proposition}{Proposition}
\newtheorem{definition}{Definition}

\newtheorem{corollary}{Corollary}

\begin{document}

\title{ Copula-Based Modeling of RIS-Assisted Communications: Outage Probability Analysis}

\author{ Im\`ene~Trigui,  Member, \textit{IEEE},  Damoon Shahbaztabar,  Member, \textit{IEEE},\\ Wessam Ajib, Senior Member, \textit{IEEE}, and  Wei-Ping Zhu, Senior Member, \textit{IEEE}.
\vspace{-0.5cm}
\thanks{

Im\`ene~Trigui and Wessam Ajib are with the Departement d’informatique, Université du Québec à Montréal, Montreal, QC H2L 2C4, Canada, e-mail: trigui.imene@uqam.ca and ajib.wessam@uqam.ca.
Damoon Shahbaztabar and Wei-Ping Zhu are with the Department of Electrical and Computer Engineering, Concordia University, Montreal, QC H3G 1M8, Canada, e-mail:damoon.shahbaztabar@concordia.ca and weiping@ece.concordia.ca. }
}

\maketitle
\begin{abstract}
Statistical characterization of the signal-to-noise ratio (SNR) of reconfigurable intelligent surface (RIS)-assisted communications in the presence of phase noise is an important open issue.  In this letter, we  exploit the concept of copula modeling to capture the non-standard dependence features that appear due to  the presence of discrete  phase noise. In particular, we consider the outage probability of RIS systems in Rayleigh fading channels and provide  joint distributions to characterize the dependencies due to the use of finite resolution phase shifters at the RIS.
 Numerical assessments confirm the validity of closed-form expressions of the outage probability and motivate the use of bivariate copula for further RIS studies.

\end{abstract}

\vspace{-0.3cm}
\section{Introduction}
Reconﬁgurable intelligent surfaces (RISs) have recently received remarkable attention as a revolutionary technology for the next generation of wireless communications to provide higher quality of service and spectrum efficiency \cite{Liu}.   It consists of arrays  of passive reflecting elements able to introduce specific phase shifts on the impinging signal \cite{Renzo}, \cite{Wu}, \cite{trigui}. Attracted by the
 appealing advantages of RIS, most works focused on the phase shift matrix design with/without the transmit beamforming to achieve optimal performance and maximum reliability (\!\!\cite{Liu} and references therein). The influence of phase noise due to low-resolution quantization and imperfect channel estimation has also been investigated in  \cite{trigui}-\!\!\cite{ITrigui}.   In particular, a power scaling analysis in  \cite{qw} showed that  using 2 or 3-bit phase shifters is practically sufficient to achieve close-to-optimal performance.  However,  there are very few works that  derived the outage probability of RIS-aided communication systems in the presence of phase noise.   To name a few, the central limit theorem (CLT) was utilized  to  approximate the RIS channel as a point-to
point Nakagami fading channel in \cite{Badiu}.  In \cite{Wang}, the authors  presented a rough asymptotic outage approximation under the consideration of one bit phase quantization and the non-realistic assumption of perfect independence  between the signal components.
Despite the large effort, however, the  exact outage probability
considering  $b$-bit phase quantization  is still not available in  the open literature.
 The paucity of this investigation is mainly due to the  mathematical  intricacy  of  handling  the  cascaded  fading channels of the RIS the  existence of non-standard
dependence features between the received signal components  caused by discrete  phase noise. To address this open issue, we  utilize the concept of copula modeling to study the  impact of the joint distribution of  the outage probability. In fact, copulas   allow modeling general
dependency structures and have already been used in
the area of communications \cite{Nelsen}, \cite{Bickel},  \cite{Ghadi}\\
To the best of our knowledge, there has been no previous work applying the copula theory to investigate the performance of RIS-assisted communications.  In the previous works, either  CLT analysis \cite{Badiu} or asymptotic formulation \cite{Zhi}, \cite{Wang} are suggested for RIS with phase noise. However, the channel models proposed in the absence of phase noise have to be further simplified to tractable formulation as shown in for example \cite{u1}, where the composite channel gain is approximated by the Gamma distribution. Unlike previous works, we use copula modeling to realize the non-linear dependence structure that appears due to phase noise.
As  such,  we  characterize  the  joint  distribution  and  propose  tractable  model  based  on  copulas.  Using  copula  model, we  derive  closed-form  outage probabilities considering general $b$-bit phase quantization.  Our results  provide a solid basis for future studies of and system design of RIS-assisted networks.\vspace{-0.32cm}
\section{System Model}
In this paper, we consider an RIS-aided
system, which consists of a single-antenna transmitter,
an RIS equipped with $M$ elements, and a single-antenna receiver. The  RIS dynamically adjusts the reflecting coefficient of each element to reconfigure the incident signal with the desired phase shift. Thus, the received signal
at the receiver is written as\vspace{-0.15cm}
\begin{equation}
  y=\sqrt{ l p_t }{\bf h}{\bf\Phi}{\bf g}s+{ z}_1,
    \label{eq1}
\end{equation}
where ${ z}_1$ represents thermal noise with power $\sigma^{2}$, $p_t$ is the transmit power, and $l$ is the equivalent
path-loss of the RIS link, which
is composed of a forward channel from the transmitter to the RIS
and a backward channel from the RIS to the receiver, denoted by
${\bf h}\in {\cal C}^{M\times 1}$ and ${\bf g}\in {\cal C}^{1\times M}$, respectively. We assume that all
the channels in the system undergo independent Rayleigh fading, and the instantaneous channel state information (CSI) for all links is assumed to be available at the receiver and RIS.
Hence, the RIS is expected to intelligently reconfigure the wireless channel by varying a quantized phase matrix ${\bf\Phi}$ defined as\vspace{-0.15cm}
\begin{equation}
    {\bold \Phi}=\Big\{e^{j\phi_n}, \phi_n\in\Big\{0, \frac{2\pi}{L}, \ldots,\frac{2\pi(L-1)}{L}\Big\}\Big\},
\end{equation}
where $L=2^{b}$, with $b\geq 1$, is the
number of discrete phases that can be generated by the RIS subject
to hardware complexity and power consumption \cite{ITrigui},\cite{trigui}. \\
In practice, the phase shifts of the reconfigurable elements of an RIS cannot be optimized with an arbitrary precision because of the finite number of quantization bits used or possible errors in estimating the phases of fading channels. In this case, the phase of the $i$-th element of the RIS can be written as $\phi_i=-\angle{h_i}-\angle{g_i}+\theta_i$, where $\theta_i$ denotes a random phase noise, which is assumed to be i.i.d. in this paper. Thus, the equivalent channel observed by the  receiver is a complex random variable and the SNR is \vspace{-0.25cm}
\begin{equation}
\gamma=\rho\left|\sum_{i=1}^{M} |h_i| |g_i| e^{j \theta_i}\right|^{2},
\label{snoise}
\end{equation}
where $\rho=l \rho_S$  with $\rho_S=\frac{ p_t}{\sigma^{2}}$ denoting the transmit SNR and  $\theta_i$
represents the phase error which is uniformly
distributed over $[-\pi/L,\pi/L]$.  To the best of our knowledge, characterising the distribution of (3) as the key to the outage analysis of the RIS-aided system is not straightforward, and it is the first endeavor in this paper.\vspace{-0.2cm}
\section{A Brief Review of Copula Theory}
Copula as a novel method to model correlated random
variables (RVs) \cite{Nelsen}, enables the computation of the joint distributions of these RVs from their marginal PDFs.
Each copula function is defined by a particular dependence
parameter which indicates the intensity of dependency.
\begin{definition}(Copula) A copula is an $n$-dimensional distribution function with standard uniform marginals.
\end{definition}
The practical relevance of copulas stems from Sklar's theorem, which we restate in the following.
\begin{theorem} (Sklar's Theorem \cite{Nelsen}). Let $F(x_1,x_2,\ldots,x_n)$ be an $n$-dimensional joint cumulative distribution function of random variables $(x_1,x_2,\ldots,x_n)$ with marginal CDFs  $F_1(x_1),\ldots,F_n(x_n)$.  Then there exists a copula $C$ such that,\vspace{-0.1cm}
\begin{equation}
F(x_1,\ldots, x_n) = C(F_1(x_1),\ldots, F_n(x_n)),
\label{F1}
\end{equation}
for all $x_i \in \mathbb{R}$. Furthermore, if $F(x_i)$ is continuous for all $i = 1, \dots, n$, then $C$ is unique.
\end{theorem}
Although many types of copulas have been defined so
far, we exploit the  Farlie Gumbel-Morgenstern (FGM) copula function [12] to analyze the performance metrics of the considered RIS-assisted system.  The generalized FGM copula of $n$-dimension is defined as \vspace{-0.2cm}
\begin{eqnarray}
    \!\!\!\!\!\!\!\!C(u_1, \ldots,u_n)\nonumber \\&&\!\!\!\!\! \!\!\!\!\!\!\!\!\!\!\!\!\!\!\!\!\!\!\!\!\!\!\!\!\!\!\!\!\!\!\!\!\!\!\!\!\!\!\!\!\!=u_1 u_2\ldots u_n \left(\!\!1+\sum_{k=2}^{n}\sum_{1 \leq j_1 \leq \ldots \leq j_n}\!\!\!\!\!\!\!\!\Theta_{j_1j_2\ldots,j_n}{\tilde u}_1{ \tilde u}_2\ldots{\tilde u}_n\!\! \right),
    \label{c1}
\end{eqnarray}
where ${\tilde u}=1-u$ and $\Theta \in [-1, 1]$ is a dependence structure parameter of the FGM copula.
\vspace{-0.27cm}
\section{Outage Probability Analysis Based on FGM Copula}
Here, we provide an FGM copula based  construction of bivariate random variables with arbitrary dependency. This construction is then used to derive the outage probability (OP) and  unveil the impact of the
joint distribution and the dependency on its performance.\\
The OP is defined as the probability that the instantaneous SNR $\gamma$ fall below a determined threshold $\gamma_{th}$, namely, \vspace{-0.22cm}
\begin{equation}
{\cal O}_p\!=\!{\rm P}\left(\!\rho\left|\sum_{i=1}^{M} |h_i| |g_i| e^{j \theta_i}\right|^{2}\!\!\leq\!\gamma_{th}\!\right)={\rm P}\!\left( \! {X^2}\! +\! {Y^2}\!\leq\! \rho_t\!\right),
\label{snr}
\end{equation}
where $\rho_t={\gamma_{th}}/{\rho}$ ,  $X=\sum_{i=1}^{M}|h_i| |g_i|\cos(\theta_i)$ and $Y=\sum_{i=1}^{M}|h_i| |g_i|\sin(\theta_i)$.
In (\ref{snr}), the real part and the imaginary part of the received signal through the RIS are separated.
We note that $X$ and $Y$ exhibit  generalized dependence structures beyond the
simple linear correlation concept widely used in wireless communications. Moreover,  considering the randomness of $X$ and $Y$ and using the transformation of random variables, the outage probability can be formulated as \vspace{-0.15cm}
  \begin{equation}
       {\cal O}_p =\int_{0}^{\rho_t}\int_{0}^{\sqrt{\rho_t-x}}\frac{1}{\sqrt{x}}f_{X,Y}\left(\sqrt{x}, y\right) dx dy,
       \label{opI}
   \end{equation}
where $f_{X,Y}\left(x, y\right)$ is the joint PDF of $X$ and $Y$.  Next, we exploit the FGM Copula to construct the joint distribution $f_{X,Y}\left(x, y\right)$ .

\begin{proposition} The copula-based joint distribution of $X$ and $Y$ under one-bit quantization is obtained as\vspace{-0.15cm}
\begin{eqnarray}
    f_{X, Y}(x,y)&=&\frac{e^{-x-|y|}x^{M-1}}{2^{M}\Gamma(M)^{2}}\nonumber \\ &&\!\!\!\!\!\!\!\!\!\!\!\!\!\!\!\!\!\!\!\!\!\!\!\!\!\!\!\!\!\!\!\!\!\!\!\!\!\!\!\!\!
    \times\sum_{k=0}^{M-1}\frac{(M-1+k)!|y|^{M-1+k}}{2^{k}k!(M-1-k)!}
    \Bigg(1+\Theta\left(1-\frac{2\Gamma(M,x)}{\Gamma(M)}\right)\nonumber \\ && \!\!\!\!\!\!\!\!\!\!\!\!\!\!\!\!\!\!\!\!\!\!\!\!\!\!\!\!\!\!\!\!\!\!\!\!\!\!\!\!\!\times\left(1-\frac{1}{2^{M-1}\Gamma(M)}\!\!\sum_{k=0}^{M-1}\!\!\frac{(M-1+k)!\Gamma(M-k,y)}{2^{k-1}k!(M-1-k)!}\right)\!\Bigg),
    \label{fxy}
\end{eqnarray}
where $\Theta\in[-1,1]$ and  $\Gamma(a,z)$ stands for the incomplete Gamma function \cite{grad}.
\end{proposition}
\begin{proof}
Using the copula theory and (\ref{F1}), the corresponding joint PDF can be obtained as follows:
\begin{eqnarray}
f_{X,Y}(x,y)&=&\frac{d^{2}C(F_X(x),F_Y(y))}{dx dy}\nonumber \\ &\overset{(a)}{=}&
 c(F_X(x),F_Y(y))f(x)f(y),
 \label{pdf}
\end{eqnarray}
where $(a)$ follows  by utilizing the concept of Chain rule with $f_X(x)$ and $f_Y(y)$ being the marginal pdfs of $X$ and $Y$, respectively, and $c(F_X(x), F_Y(y))$ denotes the bivariate Copula density function obtained from (\ref{c1}) as
\begin{equation}
   \!\!\! c(F_X(x), F_Y(y))=1+\Theta\left(\left(2F_X(x)\!-\!1\right)\left(2F_Y(y)\!-\!1\right)\right).
    \label{cf}
\end{equation}
 We assume that each element of the RIS is a one-bit phase shifter, then considering  the phase errors $\theta_i$, $i \in \{1, \ldots, M\}$ are mutually independent and uniformly
distributed on the interval $[$-$\pi/2, \pi/2]$, the marginal distributions of $X$ and $Y$ can be obtained, by following the rationale presented in Appendix A,  as
\begin{equation}
    f_X(x)=\frac{e^{-x}x^{M-1}}{\Gamma(M)}, \quad x\geq 0,
    \label{px}
\end{equation}\vspace{-0.15cm}
and \vspace{-0.2cm}
\begin{equation}
    f_Y(y)=\frac{e^{-|y|}}{2^{M}\Gamma(M)}\sum_{k=0}^{M-1}\frac{(M-1-k)!|y|^{M-1+k}}{2^{k}k!(M-1-k)!}, \quad y\in \mathbb{R}\label{pq}
\end{equation}
while the CDFs of $X$ and $Y$ can be obtained as
\begin{equation}
    F_X(x)=1-\frac{\Gamma(M,x)}{\Gamma(M)},\label{py}
\end{equation}
 and\vspace{-0.2cm}
\begin{equation}
   \!\!\! F_Y(y)=1-\frac{1}{2^{M}\Gamma(M)}\sum_{k=0}^{M-1}\frac{(M-1+k)!\Gamma(M-k,y)}{2^{k}k!(M-1-k)!}.\label{pz}
    \end{equation}
Hence, plugging (\ref{py})-(\ref{pz})  back into (\ref{cf}) and  then substituting (\ref{cf})-(\ref{pq}) into (\ref{pdf}) complete the proof.
\end{proof}
\begin{proposition}
When the RIS uses one-bit phase shifters, the outage probability in Rayleigh fading is given by (15), shown at the top of this page,
  \begin{figure*}\vspace{-0.25cm}
  \begin{eqnarray}
 {\cal O}_p&=&\frac{\rho_t^{M/2}(1+\Theta)}{2\sqrt{\pi}\Gamma(M)}{\rm H}_{0,1:1,1,1,2}^{0,0:1,1,2,1} \Bigg[\begin{array}{ccc}\sqrt{\rho_t} \\  \sqrt{\rho_t}\end{array}\Bigg|\begin{array}{ccc}  -: (1\!-\!\frac{M}{2}, \frac{1}{2});  (0,2) \\(-\frac{M}{2};\frac{1}{2},2)\!:\!(0,1);(0,2), (M-\frac{1}{2},2) \end{array}\Bigg.\Bigg]\nonumber \\ && -
 \frac{\rho_t^{M/2}\Theta}{2\sqrt{\pi}\Gamma(M)^{2}}{\rm H}_{0,1:2,2,1,2}^{0,0:2,2,2,1} \Bigg[\begin{array}{ccc}\sqrt{\rho_t} \\  \sqrt{\rho_t}\end{array}\!\! \Bigg|\begin{array}{ccc}  -:\! (M, 1), (1-\frac{M}{2},\frac{1}{2}); (0,2) \\(-\frac{M}{2};\frac{1}{2},2):(0,1), (M,1);(0,2), (M\!-\!\frac{1}{2},2) \end{array}\Bigg. \Bigg]
  \nonumber \\ && + \frac{\Theta\rho_t^{M/2}}{2^{3M}\Gamma(M)^{2}} \!\!\! \sum_{k=0}^{M-1}\sum_{t=0}^{M-1}\sum_{p=0}^{M-t-1}\frac{\Gamma(M\!+\!k)\Gamma(M\!-\!t)}{2^{2k+t+p}k!t!p!\Gamma(M\!-\!k)}\Bigg(\!\!{\rm H}_{0,1:1,1,1,2}^{0,0:1,1,2,1} \Bigg[\!\!\!\begin{array}{ccc}\sqrt{\rho_t} \\  2\sqrt{\rho_t}\end{array} \!\!\Bigg|\begin{array}{ccc}  -: (1-\frac{M}{2}, \frac{1}{2});  (1,1) \\(0;\frac{1}{2},\frac{1}{2}):(0,1); (M\!+\!k\!+\!p,1), (0,1)\end{array}\Bigg.\Bigg] \nonumber \\ &&
 \!-\!2{\rm H}_{0,1:2,2,1,2}^{0,0:2,2,2,1} \Bigg[\begin{array}{ccc}\!\!\!\!\sqrt{\rho_t} \\  \!\!2\sqrt{\rho_t}\!\!\end{array}\!\!\Bigg|\begin{array}{ccc} -: (M,1),(1\!-\!\frac{M}{2}, \frac{1}{2}); (1,1) \\(0;\frac{1}{2},\frac{1}{2})\!:\!(0,1), (M,1);(0,1), (M\!+\!k\!+\!p,1) \end{array}\Bigg. \Bigg] \Bigg),
  \end{eqnarray}
   \hrulefill\vspace{-0.35cm}
  \end{figure*}
  where $H[\cdot, \cdot]$ is the bivariate Fox’s H-function \cite[ Eq. (2.56)]{mathai}, for which several efficient implementations   have been reported in the literature \cite{trigui}.
  \end{proposition}
\begin{proof}
Using the copula-based joint distribution in Proposition 1 and resorting to $\Gamma(n,x)=\Gamma(n) e^{-x}\sum_{k=0}^{n-1}\frac{x^{k}}{k!}$, the inner integral in (\ref{opI}) can be evaluated, yielding  \begin{eqnarray}
    {\cal O}_p\!\!\!&=&\!\!\!\int_{0}^{\rho_t}\frac{f_X(\sqrt{x}){\cal K}(\Theta, x)}{2\sqrt{\pi}\Gamma(M)\sqrt{x}}  {\rm G}_{0,2}^{2,0} \Bigg[\frac{\rho_t\!-\!x} {4}\Bigg|\begin{array}{ccc}  - \\0, M\!-\!\frac{1}{2} \end{array}\Bigg. \!\!\Bigg]dx\nonumber \\ &&\!\!\!\!\!\!\!\!\!\!\!\!\!\!\!+ \int_{0}^{\rho_t}\frac{f_X(\sqrt{x})}{\sqrt{x}} \left({\cal K}(\Theta, x )-1\right) B(\sqrt{\rho_t-x})dx,
     \label{I}
     \end{eqnarray}
where ${\cal K}(\Theta, x)=1+\Theta\left(2F_X(\sqrt{x})-1\right)$, and
\begin{eqnarray}
    \!\!\! \!B(z)&=& \frac{1}{2^{3M}\Gamma(M)^{2}}  \sum_{k=0}^{M-1}\sum_{t=0}^{M-1}\sum_{p=0}^{M-t-1}\nonumber \\ && \!\!\!\!\!\!\!\!\!\!\!\!\!\!\!\!\!\!\!\!\!\!\!\!\!\!\!\!\!\!\!\frac{(M-1+k)!(M-1-t)!}{2^{2k+t+p}k!t!p!(M-1-k)!}
     {\rm G}_{1,2}^{1,1} \Bigg[2 z \Bigg|\begin{array}{ccc}  1 \\M+k+p,0 \end{array}\Bigg. \!\!\Bigg],
     \end{eqnarray}
     with  $ {\rm G}(\cdot)$ being the Meijer's G function \cite{grad}.
Now substituting (\ref{px})-(\ref{py}) into (\ref{I}), recognizing that  $\exp(-\sqrt{x})=\frac{1}{2\pi j}\int_{{\cal L}}\Gamma(s) x^{-s/2} ds$, $\Gamma(a,\sqrt{x})=\frac{1}{2\pi j}\int_{{\cal L}}\frac{\Gamma(a+s)\Gamma(s)}{\Gamma(1+s)} x^{-s/2} ds$ and \cite[Eq. (9.301)]{grad}  and then utilizing \cite[Eq.(3.194.1)]{grad}, the outage probability follows from applying \cite[Definiton A.1]{mathai}, which completes the proof.
\end{proof}
  \begin{corollary}
   The asymptotic (for high-SNR) outage probability of an RIS-aided system under one-bit quantization can be formulated as ${ \cal O}_p\underset{\rho\rightarrow\infty}{\approx}({\cal G}_c \rho )^{-{\cal G}_d}$,
where ${\cal G}_d=\frac{M}{2}$ denotes the diversity order and ${\cal G}_c=\frac{\Gamma(M-\frac{1}{2})}{2\sqrt{\pi}(1+\frac{M}{2})\Gamma(M)}$ denotes the coding gain.
  \end{corollary}
\begin{proof}
It follows  by using the asymptotic expansion of Mellin-Barnes integrals of the bivariate Fox's H functions in (15) near zero  by applying \cite[Eq. (1.8.4]{kilbas} and by keeping only the dominant terms using \cite[Eq. (1.8.7)]{kilbas}.
\end{proof}
In the most general case, i.e., for an arbitrary choice of quantization level $L$ and generalized  fading model, the marginal densities and distributions of $X$ and $Y$ are either unknown or expressed in terms
of infinite integrals \cite[Appendix A]{ITrigui}. In this case, the copula-based joint density become much more complex. To circumvent this problem, we combine, hereafter, both copula and Gamma modeling.

\begin{lemma}
Letting $Z\in\{X, Y\}$, the distribution of $Z^{2}$ can be accurately approximated by the Gamma PDF with shape and scale parameters
given by $\kappa_Z$ and $\beta_Z$  as\vspace{-0.2cm}
\begin{equation}
    Z^{2}\simeq \Gamma(\kappa_Z, \beta_Z),
\end{equation}\vspace{-0.22cm}
where $\kappa_Z=\frac{\mathbb{E}(Z^{2})^{2}}{\mathbb{E}(Z^{4})-\mathbb{E}(Z^{2})^{2}}$ and $\beta_Z=\frac{\kappa_Z}{\mathbb{E}(Z^{2})}$.
\end{lemma}

\begin{proof}
In order to approximate $Z^{2}, Z\in\{X, Y\}$ being a Gamma random variable,  we have to find the shape and scale parameters
(i.e., $\kappa_Z$, $\beta_Z$) based on the statistical information of $Z$. To this end, we use  two different moments of $Z$ i.e.,  $\mathbb{E}(Z^{2})$ and $\mathbb{E}(Z^{4})$ to find the parameters $\kappa_Z$ and $\beta_Z$. First, from the definitions of $X$ and $Y$, we have \vspace{-0.24cm}
\begin{equation}
   X^{2}=\sum_{i=1}^{M}z_i^{2}\cos(\theta_i)^{2}+\sum_{i=1}^{M}\sum_{j=1, j\neq i}^{M}z_iz_j \cos(\theta_i)\cos(\theta_j),
\end{equation}\vspace{-0.25cm}
and \vspace{-0.15cm}
\begin{equation}
   Y^{2}=\sum_{i=1}^{M}z_i^{2}\sin(\theta_i)^{2}+\sum_{i=1}^{M}\sum_{j=1,j\neq i}^{M}z_iz_j \sin(\theta_i)\sin(\theta_j),
\end{equation}
where $z_i=|h_i||g_i|$.  Since $|h_i|$ and $|g_i|$ undergo i.i.d. unit variance Rayleigh fading, we have   $\mathbb{E}(z_i)=\frac{\pi}{4}$, $\mathbb{E}(z_i^{2})=1$,  $\mathbb{E}(z_i^{3})=\frac{9\pi}{16}$, and  $\mathbb{E}(z_i^{4})=4$. Moreover, since $\theta_i$ is a uniform random variable with PDF $f_{\theta}(x)=\frac{L}{2\pi}$, $-\frac{\pi}{L}\leq x \leq\frac{\pi}{L}$, we obtain \vspace{-0.25cm}
\begin{equation}
    \mathbb{E}( Z^{2})=\begin{cases}
M\left(\frac{1}{2}\!+\!\frac{1}{2}\rm{sinc}\left(\frac{2\pi}{L}\right)\!+\!\frac{(M-1)L^{2}}{16}\sin^{2}\left(\frac{\pi}{L}\right) \right) & Z=X;\\M\left(\frac{1}{2}-\frac{1}{2}\rm{sinc}\left(\frac{2\pi}{L}\right)\right) & Z=Y,
\end{cases}
\label{moy1}
\end{equation}
Next, in order to find the moment $\mathbb{E}( Z^{4})$, we initially expand it as\vspace{-0.3cm}
\begin{eqnarray}
    \mathbb{E}( Z^{4})&=&\mathbb{E}\left[\left(\sum_{i=1}^{M}d_i^{2}\right)^{2}\right]+ 2 \sum_{l=1}^{M}\sum_{i=1}^{M}\sum_{j=1, j\neq i }^{M}\mathbb{E}\left[d_l^{2} d_id_j\right]\nonumber \\ && +\mathbb{E}\left[\left(\sum_{i=1}^{M}\sum_{j=1, i\neq j }^{M}d_id_j\right)^{2}\right],
    \label{moy}
\end{eqnarray}
where $d_i=z_i \cos(\theta_i)$ if $Z=X$ or $d_i=z_i \sin(\theta_i)$ if  $Z=Y$. The first term on the rights side of (\ref{moy}) can be expressed as \vspace{-0.25cm}
\begin{eqnarray}
\mathbb{E}\left[\left(\sum_{i=1}^{M}d_i^{2}\right)^{2}\right]=\sum_{i=1}^{M}\mathbb{E}\left[d_i^{4}\right]+\sum_{i=1}^{M}\sum_{j=1,j\neq i}^{M}\mathbb{E}\left[d_i^{2}d_j^{2}\right],
\end{eqnarray}
which, after some manipulations, can be rewritten as \vspace{-0.25cm}
\begin{eqnarray}\mathbb{E}\left[\left(\sum_{i=1}^{M}d_i^{2}\right)^{2}\right]=
\begin{cases}\frac{M}{16 \pi^{2}}(A\!+B) & Z=X;\\\frac{M}{16 \pi^{2}}(B-A) & Z=Y,\end{cases}
\label{e1}
\end{eqnarray}
where $A=4 \pi(3\!+\!M)L\sin\left(\frac{2\pi}{L}\right)$ and $B=\!L^{2}(M\!-\!1)\sin^{2}\left(\frac{2\pi}{L}\right)\!+\!2 \pi\left(2(5\!+\!M)\pi\!+\!L\sin\left(\frac{4\pi}{L}\right)\right)$. The second term of (\ref{moy}) can be expressed as (\ref{ex1}), as
shown at the top of this page. Further, the third term of (\ref{moy}) can be expressed as (\ref{ex2}), as shown at the top of the this page. Finally, after plugging (\ref{e1}), (\ref{ex1}) and
(\ref{ex2}) back into (\ref{moy}) and using (\ref{moy1}) and  several simplifications, Lemma 1 is proved.
\begin{figure*}\vspace{-0.3cm}
  \begin{equation}
 \begin{array}{l}
2\sum\limits_{l = 1}^M {\sum\limits_{m = 1}^M {\sum\limits_{n = 1, n\neq m}^M {{\mathbb E}\left[ {d_l^2{d_m}{d_n}} \right]}} }  = 4\sum\limits_{m = 1}^M {\sum\limits_{n = 1,n \ne m}^M {{\mathbb E}\left[ {d_m^3{d_n}} \right]} }
 + 2\sum\limits_{l = 1}^M {\sum\limits_{m = 1,m \ne l}^M {\sum\limits_{n = 1,n \ne m}^M {{\mathbb E}\left[ {d_l^2{d_m}{d_n}} \right]} } } \\
 =\begin{cases}\frac{1}{{64 \pi}}M\left( {M - 1} \right){L^2}{\sin ^2}\left( {\frac{\pi }{L}} \right)\left((11+4M)\pi+3 \pi \cos\left( {\frac{2\pi }{L}} \right)+2L(M-1)\sin\left( {\frac{2\pi }{L}} \right)\right) & Z=X;\\
 0& Z=Y,
 \end{cases}
\end{array}
\label{ex1}
\end{equation} \vspace{-0.2cm}
 \hrulefill
\end{figure*}
\begin{figure*}\vspace{-0.25cm}
 \begin{equation}
\begin{array}{l}
{\mathbb E}\left[ {{{\left( {\sum\limits_{m = 1}^M {\sum\limits_{n =  1, n\neq m}^N {{d_m}{d_n}} } } \right)}^2}} \right]
 = \sum\limits_{j = 1}^M {\sum\limits_{l =1, l \ne m}^M { {{\mathbb E}\left[ {d_j^2d_l^2} \right]} } }
 + \sum\limits_{l = 1}^M {\sum\limits_{m = 1,m \ne l}^M {\sum\limits_{n = m + 1,n \ne l}^N {{\mathbb E}\left[ {d_l^2{d_m}{d_n}} \right]} } } \\ + \sum\limits_{j = 1}^M {\sum\limits_{l = 1,l \ne j \ne m \ne n}^M {\sum\limits_{m = 1,l \ne j \ne m \ne n}^M {\sum\limits_{n = 1,l \ne j \ne m \ne n}^M {{\mathbb E}\left[ {{d_j}{d_l}{d_m}{d_n}} \right]} } } } \\
  =\begin{cases}\frac{M(M-1)}{256 \pi^2}\left(8 L^2 (M-1) \pi^2 \sin\left(\frac{\pi}{L}\right)^2+8\pi L^3(M-1)\sin\left(\frac{\pi}{L}\right)^3\cos\left(\frac{\pi}{L}\right)\right)  & Z=X; \\+
   L^4 \left(6  - 5 M + M^2\right) \pi^2 \sin\left(\frac{\pi}{L}\right)^4 +
   16  \left(2 \pi + L \sin\left(\frac{2 \pi}{L}\right)\right)^2 \\
 \frac{M^{2}\left(L\sin\left( {\frac{2\pi }{L}} \right)-2\pi\right)^{2}}{16\pi^{2}}& Z=Y,
 \end{cases}
\end{array}
\label{ex2}
\end{equation} \vspace{-0.25cm}
 \hrulefill
 \vspace{-0.2cm}
\end{figure*}
\end{proof}\vspace{-0.25cm}
\begin{proposition}
A copula-based joint distribution of $X^2$ and $Y^2$ under $b$-bits phase  quantization and Rayleigh fading is given by \vspace{-0.35cm}
\begin{eqnarray}
    f_{X^{2},Y^{2}}(x,y)&=&\frac{e^{-\frac{x}{\beta_X}-\frac{y}{\beta_Y}}x^{\kappa_X-1}y^{\kappa_Y-1}}{\Gamma(\kappa_X)\Gamma(\kappa_Y)\beta_X^{\kappa_X}\beta_Y^{\kappa_Y} }\nonumber\\&&\!\!\!\!\!\!\!\!\!\!\!\!\!\!\!\!\!\!\!\!\!\!\!\!\!\!\!\!\!\!\!\!\!\!\!\!\!\!\!\!\!\!\!\!\!\!\!\! \left(1+\Theta\left(1- \frac{2\Gamma(\kappa_X,\frac{x}{\beta_X})}{\Gamma(\kappa_X)}\right)\left(1- \frac{2\Gamma(\kappa_Y,\frac{y}{\beta_Y})}{\Gamma(\kappa_Y)}\right)\right).
    \label{fxy2}
\end{eqnarray}
\end{proposition}
\begin{proof}
The proof follows  in the same line of (\ref{fxy}) using the Gamma model approximation of $X^{2}$ and $Y^{2}$ in Lemma 1 after recognizing that $f_Z(z)=\frac{e^{-\frac{z}{\beta_Z}}z^{\kappa_Z-1}}{\Gamma(\kappa_Z)\beta_Z^{\kappa_Z}}$ and $F_Z(z)=1-\frac{\Gamma(\kappa_Z,\frac{z}{\beta_Z})}{\Gamma(\kappa_Z)}$, $Z\in\{X^{2},Y^{2}\}$.
\end{proof}\vspace{-0.2cm}
\begin{proposition}
The outage probability of RIS-assisted system with $b$-bit phase  quantization can be expressed as
\begin{eqnarray}
      \! {\cal O}_p\!\!\!\!&=&\!\!\!\!\!\int_{0}^{\rho_t}\!\!f_{X^{2}}(x) \left(\Theta{\cal G}(x)+1\right)\left(1-\frac{\Gamma\left(\kappa_Y,\frac{\rho_t-x}{\beta_Y}\right)}{\Gamma(\kappa_Y)} \right)dx \nonumber \\ &&  -2\Theta \int_{0}^{\rho_t}\frac{{\cal G}(x)f_{X^{2}}(x)}{\Gamma(\kappa_Y)}\left(\frac{\rho_t-x}{\beta_Y} \right)^{\kappa_Y}\nonumber \\ && \!\!\!\!\!\!\!\!\!\!\!\!\!\!\!\!\!\!\!\!\!\!\!\!\!\!\times {\rm H}_{0,1:0,1,1,2}^{0,0:1,0,2,0} \Bigg[\begin{array}{ccc}\frac{\rho_t-x}{\beta_Y}\\  \frac{\rho_t-x}{\beta_Y}\end{array}\Bigg|\begin{array}{ccc}  -: -;  (1,1) \\\!\!\!(\kappa_Y,1,1):\!(0,1);(0,1),(\kappa_Y,1) \end{array}\Bigg.\!\!\!\!\Bigg] dx,\nonumber \\
   \end{eqnarray}
   where ${\cal G}(x)=1- \frac{2\Gamma(\kappa_X,\frac{x}{\beta_X})}{\Gamma(\kappa_X)}$.
\end{proposition}
\begin{proof}
It follows from substituting (\ref{fxy2}) into (\ref{opI}) and solving the integral with respect to $y$ by following similar steps as in  Proposition 2.
\end{proof}

To the best of our knowledge,  Propositions 1 and 3, characterize for the fisrt time in the literature the joint distribution between the underlying real and imaginary parts of the received signal through RIS, which due to the presence of phase noise, exhibit arbitrary correlation. This finding allows us to derive the  exact outage probability for different quantization levels, as shown in Proposition 2 and 4 for the first disclosure.\\
It is worth noting that   other performance metrics, including the ergodic capacity, error probability and secrecy rate, also depend remarkably on the joint distributions in Propositions 1 and 2. It is therefore of interest to investigate in future research how such dependency structure can be exploited in RIS-aided communications. Moreover,  by leveraging fundamental results from the  Mellin transform  \cite{mathai} and Copula  \cite{Bickel} theories, it is possible to extend the current framework to  deal with inherent complexities due to generalized  fading models such as general multi-path with/without specular component (LOS) and shadowing.  \vspace{-0.31cm}
   \section{Numerical Results}
In Fig.~1 (a), we show the outage probability  for the
RIS-aided system  under the  condition  of one-bit phase quantization as described in Proposition 2. It can be observed that the Monte-Carlo simulation of the outage probability
 matches the copula-based closed-form expression in Proposition 2. Notice that our  approach outperforms the Nakagami-m approximation in \cite{Badiu}, which is degraded for low outage values.  Fig.~1 (a)  further illustrates that the copula-based asymptotic analysis
is very accurate showing that the diversity order is indeed $M/2$.

In Fig.~1 (b), the behavior of the outage probability based on variations of $M$ and $b$ for selected values of the dependence parameter $\Theta$ is illustrated\footnote{In general, the appropriate value of the dependence factor $\Theta$ for the FGM copula can be determined by minimizing particular cost functions using for instance the likelihood-based methods \cite{theta}.}. We examine the closeness between the simulated values of the outage and the
approximations based on the the Gamma model presented in Proposition 4.  Fig.~1 (b) corroborates
the fact that  increasing both $M$ and $b$ improves the performance of the RIS system. In particular, it is shown that using 2 bit phase shifters is practically sufficient to achieve close-to-optimal performance with only approximately $0.9$ dB power loss, which is consistent with \cite[Proposition 1]{qw}.

 Denoting the distances to and from the RIS as $l_1$, $l_2$,  the equivalent
path-loss of the RIS link can be written as $l= (l_1l_2)^{-\nu}$, where $\nu$ is the path-loss exponent.
In Fig.~1 (c), we investigate the outage probability for different
geographical deployments of the RIS which is allowed
to move along the horizontal line between the transmitter and the receiver. An observation from Fig.~1 (c), which agrees with the finding
in \cite{ITrigui}, is that  the RIS  has increased outage as it  moves far away from either the transmitter or the receiver. It is also noticed  from Fig.~1 (b) and (c) that copula is more accurate for positive dependence structure i.e., when  $\Theta\in(0,1]$.
\begin{figure*}\vspace{-1.4cm}
      \captionsetup[subfigure]{justification=centering}
    \centering
       \begin{subfigure}{0.31\textwidth}
        \includegraphics[width=\textwidth]{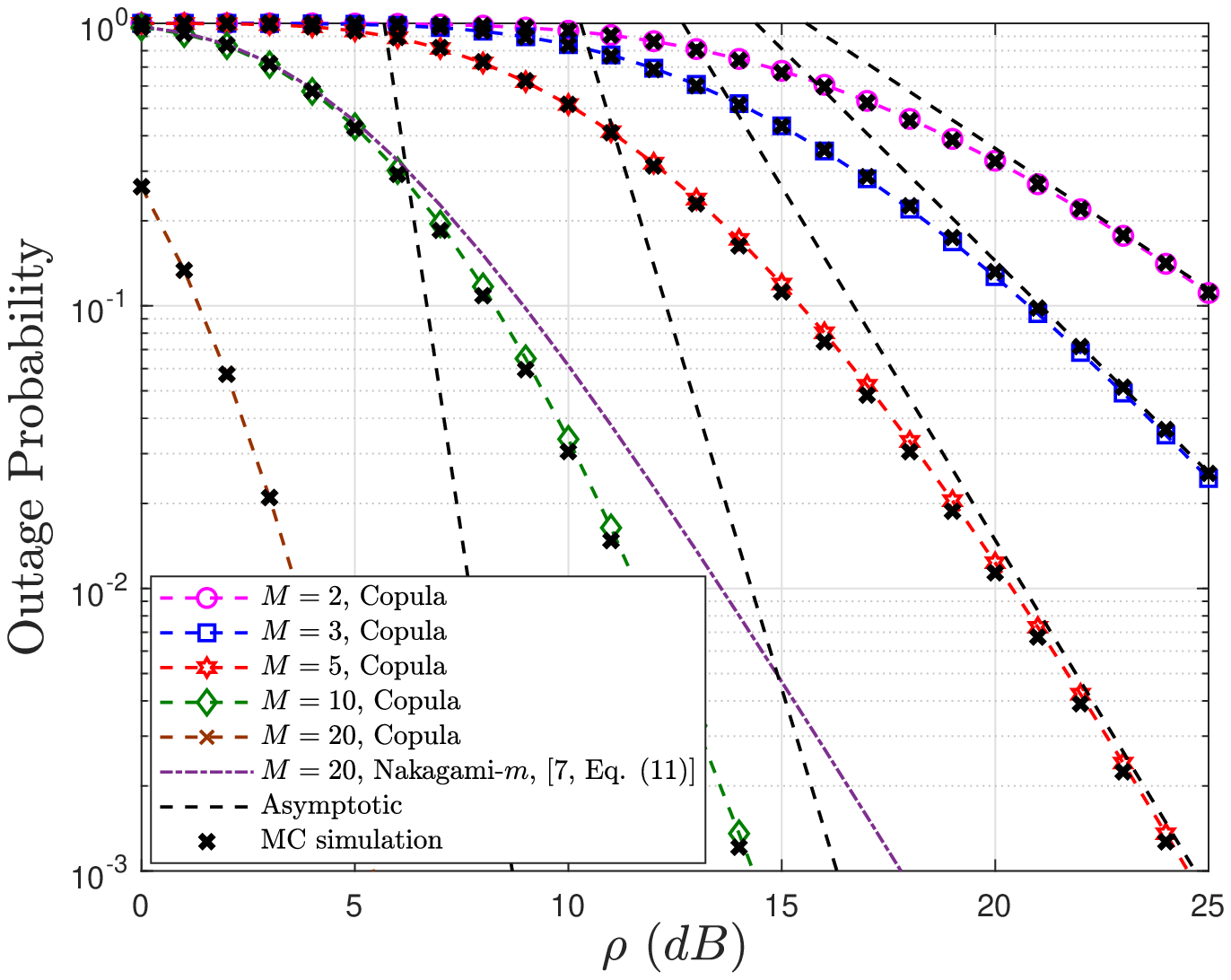}
        \caption{\\  }
    \end{subfigure}%
      \hfill \hspace{-5cm}
      \begin{subfigure}{0.31\textwidth}
        \includegraphics[width=\textwidth]{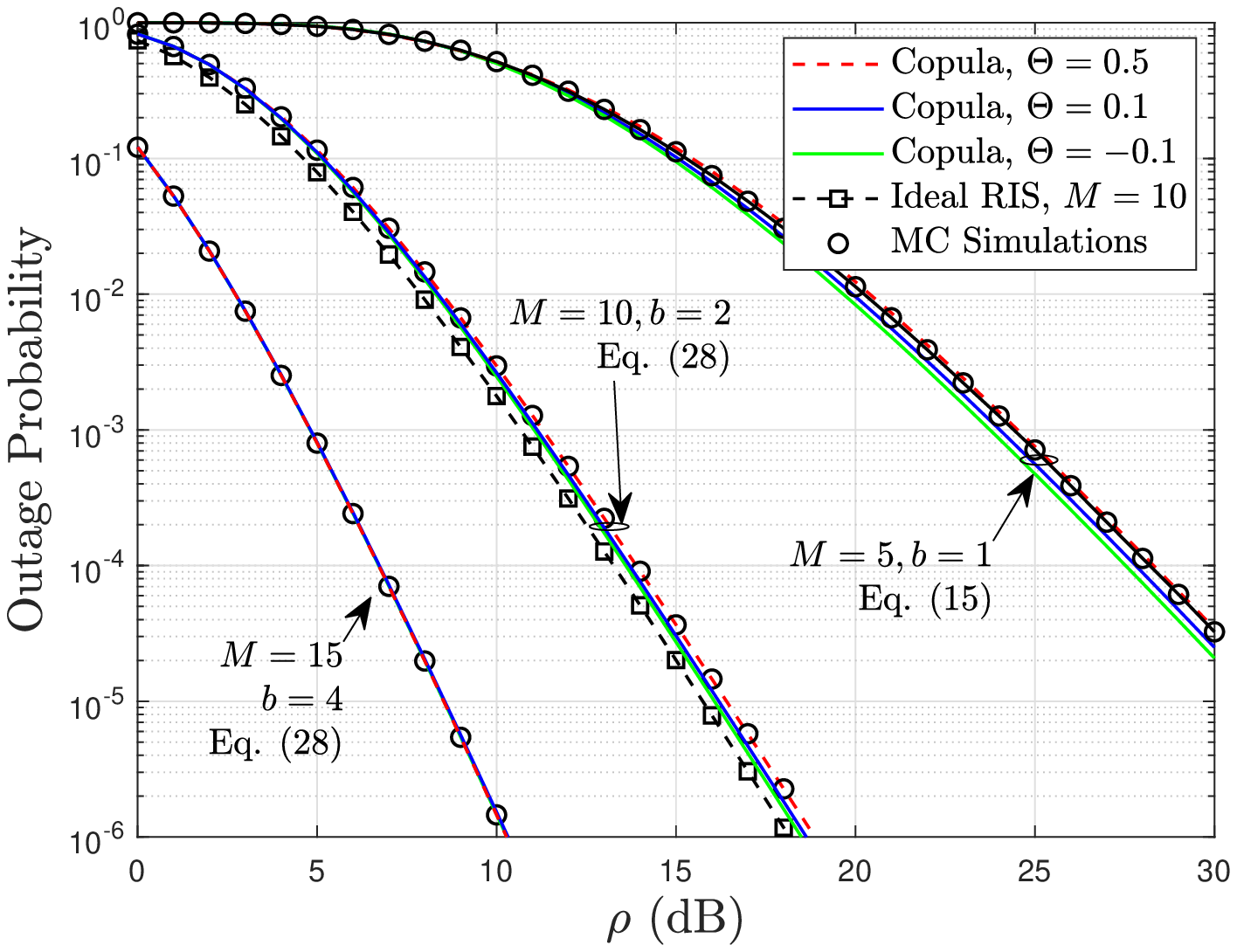}
        \caption{ \\}
    \end{subfigure}
            \hfill\hspace{-5.2cm}
      \begin{subfigure}{0.31\textwidth}
        \includegraphics[width=\textwidth]{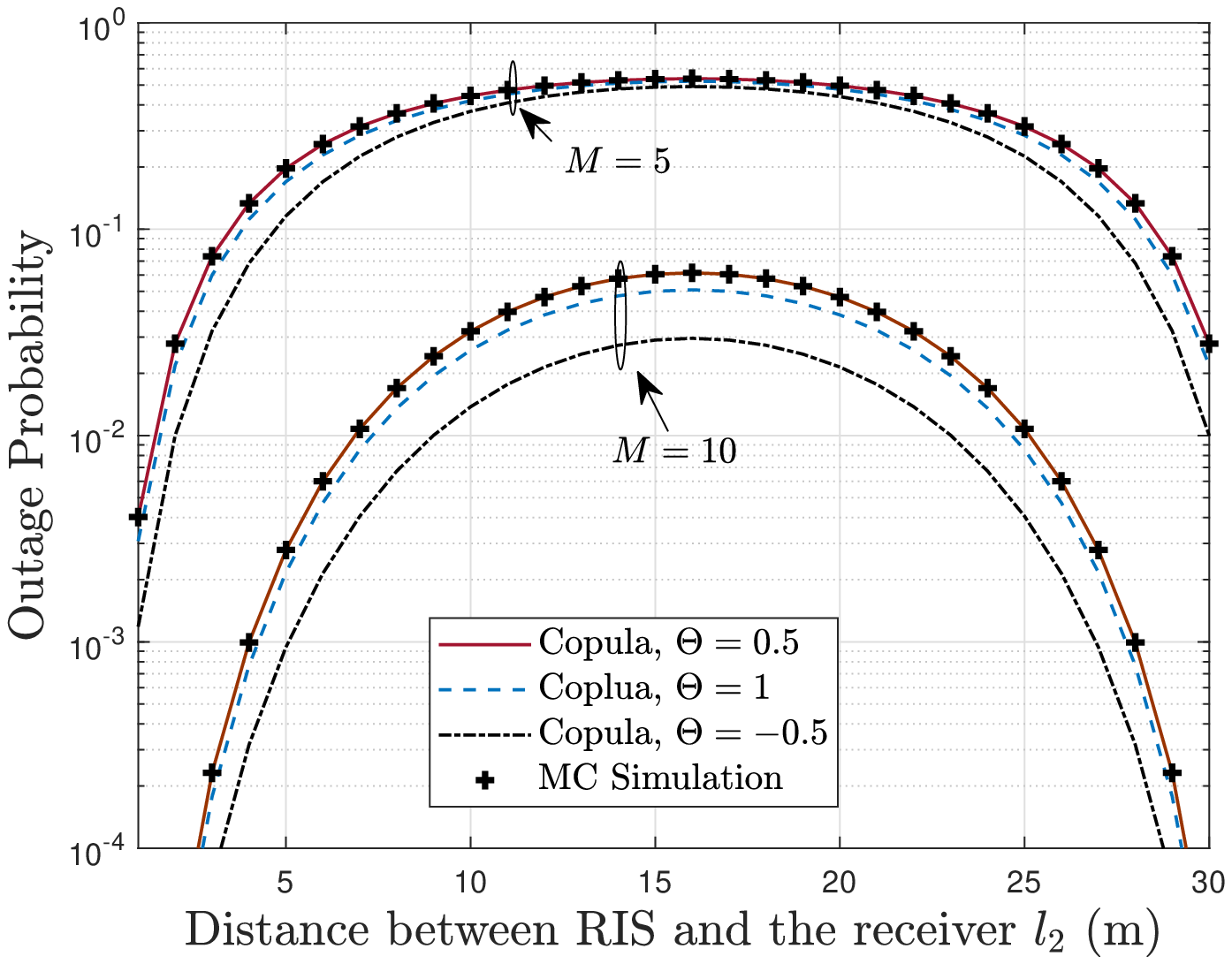}
\caption{ \\}
    \end{subfigure}\hfill\hspace{-5cm}\vspace{-0.2cm}
    \caption{\small The outage probability with $\gamma_{th}=5$ dB under (a) one-bit phase quantization with dependence parameter  $\theta=0.55$. (b)  $b$-bit quantization  for different $M$ and $\Theta$ values. (c) one-bit quantization for different positions of the RIS with $p_t/\sigma^{2}=15$ dB and $\nu=2.8$. }\vspace{-0.55cm}
                \end{figure*}

\vspace{-0.3cm}
\section{Conclusion}
In this letter, we developed a theoretical framework to analyze the outage probability of RIS-aided systems with phase noise in which copulas are exploited to capture the non-linear dependence among the signal components. When a one-bit phase shifter is used at each reflective element, we obtained the expression of the outage probability using  the bivariate Fox's H function.  To deal with the complicated scenario where the RIS employs  $b$-bit phase shifters, we amalgamated both copula and Gamma modeling to efficiently compute the outage probability.
Our results are the foundations of any further study that relies on the joint  density and cumulative
probability functions of the underlying variables pertaining to the received signal via low resolution RIS.

\vspace{-0.25cm}
\section{Appendix A}
As noted earlier,  $\theta_i$ is uniformly distributed on the interval
$[-\frac{\pi}{2}, \frac{\pi}{2}]$. As a result, the PDF of $v_i=\cos(\phi_i)$ is given by
\begin{equation}
f_{v_i}(v)=\frac{2}{\pi\sqrt{1-v^{2}}}, \quad\text{for~} 0\leq v \leq1.
\label{cs}
\end{equation}
Recall that both forward and backward channels of the RIS-aided system follow i.i.d. Rayleigh fading $f_{y_i}(x)=2 x {\rm{e}}^{-x^{2}}$, $y\in \{h,g\}$, using (\ref{cs}), the PDF  of $x_i=h_i g_i v_i$ can be calculated as\vspace{-0.2cm}
\begin{equation}
f_{x_i}(z)=\int_{0}^{\infty}\frac{1}{x}f_{h_i}(x)f_{g_i v_i}\left(\frac{z}{x}\right) \textrm{d}x,
\label{eqf}
\end{equation}
where\vspace{-0.22cm}
\begin{equation}
f_{g_i v_i}\left(x\right)=\frac{4}{\pi}\int_{0}^{1}\frac{\rm{e}^{-\frac{x^{2}}{z^{2}}}}{z^{2}\sqrt{1-z^{2}}}\textrm{d}z.
\label{eq1}
\end{equation}
Hence the pdf of $x_i$ is obtained as
$
f_{x_i}(x)=\exp(-x).
$
Similarly, the PDF of $y_i=h_i g_i\sin(\theta_i)$ is $f_{y_i}(y)=\frac{1}{2}\exp(-|y|)$. Since the sum of
independent exponentially distributed random variables with the same mean is Gamma-distributed, we can obtain easily the PDFs of  $X=\sum_{i=1}^{M}x_i$ and $Y=\sum_{i=1}^{M}y_i$   as   given in (11) and (12), respectively.

\end{document}